\documentclass[submission,copyright,creativecommons]{eptcs}
\providecommand{\event}{SOS 2007} 

\usepackage{iftex}

\ifpdf
  \usepackage{underscore}         
  \usepackage[T1]{fontenc}        
\else
  \usepackage{breakurl}           
\fi

\title{Arithmetical Characterizations of Two Intuitionistic Logical Principles}
\author{Kaito Ichikura
\institute{Tohoku University, Graduate School of Information Sciences\\
Sendai, Japan}
\email{ichikura.kaito.t7@dc.tohoku.ac.jp}
}
\def\titlerunning{Arithmetical Characterizations of Two Intuitionistic Principles}
\def\authorrunning{Kaito Ichikura}

\usepackage{amsmath}
\usepackage{amsfonts}
\usepackage{amsthm}
\usepackage{amssymb}
\usepackage{amscd}
\usepackage{wrapfig}
\usepackage{multicol}
\usepackage{enumitem}
\setlist{itemsep=0pt, parsep=0pt}
\usepackage{bussproofs}

\newcommand{\dotminus}{\mathbin{\dot{-}}}

\theoremstyle{definition}
\newtheorem{dfn}{Definition}[section]
\newtheorem{prop}[dfn]{Proposition}
\newtheorem{lem}[dfn]{Lemma}
\newtheorem{thm}[dfn]{Theorem}
\newtheorem{cor}[dfn]{Corollary}
\newtheorem{rem}[dfn]{Remark}

\begin{document}
\maketitle
	
	\begin{abstract}
    This paper investigates the role of negation in arithmetic by analyzing subsystems based on subminimal logic. We introduce $\mathbf{NA}^{PR}$, a system of subminimal arithmetic equipped with all primitive recursive functions, and answer the following four questions: (1) Can we develop reasonable arithmetic within $\mathbf{NA}^{PR}$? (2) Can we characterize logical rules Absorption of Negation ($\mathsf{An}$) and Ex Falso Quodlibet  ($\mathsf{ECQ}$) by arithmetical principles? (3) Are there strict differences in deductive power among $\mathbf{NA}^{PR}$, $\mathbf{MA}^{PR}$, and $\mathbf{HA}$? (4) Do $\mathbf{MA}^{PR}$ and $\mathbf{MA}$, and $\mathbf{NA}^{PR}$ and $\mathbf{NA}$ have a similar mutual interpretability relationship as $\mathbf{HA}$? We show the following: First, Heyting Arithmetic ($\mathbf{HA}$) is shown to be a conservative extension of $\mathbf{NA}^{PR}$ for positive formulas. Based on $\mathbf{NA}^{PR}$, we prove that $(\mathsf{ECQ})$ is equivalent to the definability of disjunction, and that $(\mathsf{An})$ is equivalent to a variant of the least number principle. Furthermore, using the constant domain neighbourhood semantics, we show that the systems $\mathbf{NA}^{PR}$, $\mathbf{MA}^{PR}$ (written as $\mathbf{NA}^{PR}+(\mathsf{An}$)) and $\mathbf{HA}$ (written as $\mathbf{NA}^{PR}+(\mathsf{An})+(\mathsf{ECQ})$) form a strictly increasing hierarchy of strength. Finally, we demonstrate the essential role of primitive recursive functions in weak arithmetic using classical countermodels.
\end{abstract}
	
	\section{Introduction}
	
	Non-classical logics, particularly in the style of sequent calculus \cite{gentzen1935untersuchungen} for intuitionistic logic ($\mathsf{LJ}$), have been widely studied to understand the constructive content of mathematical reasoning. By discarding the inference rule correspoding to Ex Falso Quodlibet ($\mathsf{ECQ}$): $A\to(\neg A\to B)$ from $\mathsf{LJ}$, we obtain sequent calculus for minimal logic  ($\mathsf{LM}$) (cf.\ \cite{johansson1937minimalkalkul}). A further weakening, which also drops the inference rule correspond to Absorption of Negation ($\mathsf{An}$): $(A \to \neg A) \to \neg A$, leads to sequent calculus for logic $\mathsf{N}$ ($\mathsf{LN}$) (cf.\ \cite[Section 3]{bilkova2020proof}, \cite{vakarelov2005nelson, bsl24752}), which is the base subminimal logic equipped only with the congruence inference rule ($\mathsf{N}$): $(A \leftrightarrow B) \to (\neg A \leftrightarrow \neg B)$ for negation.
	
	Analyzing arithmetical systems based on subminimal logic allows us to finely calibrate the role of negation in mathematical theories. By stripping negation of its powerful operational rules, we can observe which arithmetical theorems purely rely on the positive structure of numbers and which strictly require the specific logical properties of negation. We employ $\mathbf{NA}^{PR}$, which consists of $\mathsf{LN}$ and all non-logical axioms and rules of primitive recursive functions and induction, as the base theory for our investigation. Then we can obtain systems $\mathbf{MA}^{PR}$ of minimal arithmetic and $\mathbf{HA}^{PR}$ of Heyting arithmetic by adding $(\mathsf{An})$ and $(\mathsf{An})+(\mathsf{ECQ})$ to $\mathbf{NA}^{PR}$, respectively. 

	At this point, several natural questions arise:
	\begin{enumerate}
		\item Can we develop reasonable arithmetic within $\mathbf{NA}^{PR}$?
		\item It is known that over $\mathbf{HA}$, the Law of Excluded Middle ($\text{LEM}$) is characterized by the Least Number Principle ($\text{LNP}$) \cite{fujiwara2022refining}. Can analogous characterizations be possible for $(\mathsf{An})$ and $(\mathsf{ECQ})$?
		\item Are there strict differences in deductive power among $\mathbf{NA}^{PR}$, $\mathbf{MA}^{PR}$, and $\mathbf{HA}^{PR}$?
		\item It is well-known that $\mathbf{HA}^{PR}$ is mutually interpretable with $\mathbf{HA}$, which is formulated in the simpler basic language $\mathcal{L}_A$ consisting of $0, S, +, \cdot$ and $ = $ (cf.\ \cite{avigad2022mathematical}). Do $\mathbf{MA}^{PR}$ and $\mathbf{MA}$, and $\mathbf{NA}^{PR}$ and $\mathbf{NA}$ have a similar relationship?
	\end{enumerate}

	This paper answers these questions. Section \ref{sec:preliminaries} introduces the proof systems of subminimal logics and arithmetic. Addressing the first question, Section \ref{sec:basic-prop} gives some basic properties of $\mathbf{NA}^{PR}$ and proves the conservativity of $\mathbf{HA}^{PR}$ over $\mathbf{NA}^{PR}$ for positive formulas. To answer the second question, Section \ref{sec:equiv} explores equivalences between logical rules and arithmetical principles: we show that $(\mathsf{ECQ})$ is equivalent to the definability of disjunction, and we introduce the Non-contradictory Bounded Least Number Principle ($\text{NBP}$) to characterize $(\mathsf{An})$. For the third question, we develop a constant domain neighbourhood semantics, which is a predicate extension of the neighbourhood semantics in \cite[Section 4]{colacito2016minimal}. And we present countermodels separating $(\mathsf{An})$ and $(\mathsf{ECQ})$ over $\mathbf{NA}^{PR}$ in Sections \ref{sec:sep-na-ma} and \ref{sec:sep-ma-ha}. Finally, resolving the fourth question, Section \ref{sec:sep-lang} isolates the impact of primitive recursive function symbols by separating $\mathbf{NA}^{PR}$ from $\mathbf{NA}$ (the system without primitive recursive symbols) using classical countermodels.

	\section{Preliminaries}\label{sec:preliminaries}

    This section introduces our basic logical language of predicate logic, the sequent calculus $\mathsf{LN}$, and the formal definition of the subminimal arithmetical system $\mathbf{NA}^{PR}$, which serves as a base theory for our investigation.

	\begin{dfn}[Logical language and formulas]\label{dfn:language}
		The language consists of the logical connectives $\land, \lor, \to, \neg$, the quantifiers $\forall, \exists$, a countably infinite set of bound variables (typically denoted by $x, y, z, \dots$), and a countably infinite set of free variables (typically denoted by $a, b, c, \dots$), which are assumed to be mutually disjoint. 
	\end{dfn}

    Note that we do not include $\bot$ in the language and that $\neg$ is not defined by $\to$ and $\bot$.
		Formulas are constructed from atomic formulas, the logical connectives $\land, \lor, \to, \neg$, and the quantifiers $\forall, \exists$. A \textit{positive formula} is a formula that does not contain the negation connective $\neg$.
	
	\begin{dfn}[Sequents, initial sequents, and inference rules of $\mathsf{LN}$]\label{dfn:ln}
		
		A \textit{sequent} is an expression of the form $\Gamma \vdash C$, where $\Gamma$ is a finite sequence of formulas and $C$ is a formula.
		The sequent calculus $\mathsf{LN}$ consists of the following initial sequents and rules.
		
		\noindent
		\underline{Initial sequents}: {\centering $A \vdash A$}
		
		\noindent
		\underline{Structural inference rules}:
		
		\leftline{
			$\begin{array}{cc}
				\parbox{.45\linewidth}{
					\begin{prooftree}
						\AxiomC{$\Gamma\vdash C$}
						\RightLabel{$(\mathsf{weakening})$}
						\UnaryInfC{$A,\Gamma\vdash C$}
				\end{prooftree}}&
				\parbox{.45\linewidth}{\begin{prooftree}
						\AxiomC{$A,A,\Gamma\vdash C$}
						\RightLabel{$(\mathsf{contraction})$}
						\UnaryInfC{$A,\Gamma\vdash C$}
				\end{prooftree}}\\
				\parbox{.45\linewidth}{\begin{prooftree}
						\AxiomC{$\Gamma,A,B,\Sigma\vdash C$}
						\RightLabel{$(\mathsf{exchange})$}
						\UnaryInfC{$\Gamma,B,A,\Sigma\vdash C$}
				\end{prooftree}}&
				\parbox{.45\linewidth}{\begin{prooftree}
						\AxiomC{$\Gamma\vdash A$}
						\AxiomC{$A,\Sigma\vdash C$}
						\RightLabel{$(\mathsf{cut})$}
						\BinaryInfC{$\Gamma,\Sigma\vdash C$}
				\end{prooftree}}
			\end{array}$}
		
		\noindent
		\underline{Inference rule for logical connectives}:
		
		\leftline{$\begin{array}{cc}
				\parbox{.45\linewidth}{
					\begin{prooftree}
						\AxiomC{$A,B,\Gamma\vdash C$}
						\RightLabel{$(\land\mathsf{L})$}
						\UnaryInfC{$A\land B,\Gamma\vdash C$}
				\end{prooftree}}
				&
				\parbox{.45\linewidth}{
					\begin{prooftree}
						\AxiomC{$\Gamma\vdash A$}
						\AxiomC{$\Gamma\vdash B$}
						\RightLabel{$(\land\mathsf{R})$}
						\BinaryInfC{$\Gamma\vdash A\land B$}
				\end{prooftree}}
			\end{array}$}
		
		\leftline{
		$\begin{array}{ccc}
			\parbox{.32\linewidth}{
				\begin{prooftree}
					\AxiomC{$A,\Gamma\vdash C$}
					\AxiomC{$B,\Gamma\vdash C$}
					\RightLabel{$(\lor\mathsf{L})$}
					\BinaryInfC{$A\lor B,\Gamma\vdash C$}
			\end{prooftree}}
			&
			\parbox{.32\linewidth}{
				\begin{prooftree}
					\AxiomC{$\Gamma\vdash A$}
					\RightLabel{$(\lor\mathsf{R}_1)$}
					\UnaryInfC{$\Gamma\vdash A\lor B$}
			\end{prooftree}}
			&
			\parbox{.32\linewidth}{
				\begin{prooftree}
					\AxiomC{$\Gamma\vdash B$}
					\RightLabel{$(\lor\mathsf{R}_2)$}
					\UnaryInfC{$\Gamma\vdash A\lor B$}
			\end{prooftree}}
		\end{array}$}
		
		\leftline{$\begin{array}{cc}
				\parbox{.45\linewidth}{
					\begin{prooftree}
						\AxiomC{$\Gamma\vdash A$}
						\AxiomC{$B,\Sigma\vdash C$}
						\RightLabel{$(\to\mathsf{L})$}
						\BinaryInfC{$A\rightarrow B,\Gamma,\Sigma\vdash C$}
				\end{prooftree}}
				&
				\parbox{.45\linewidth}{
					\begin{prooftree}
						\AxiomC{$A,\Gamma\vdash B$}
						\RightLabel{$(\to\mathsf{R})$}
						\UnaryInfC{$\Gamma\vdash A\rightarrow B$}
				\end{prooftree}} \\
				\parbox{.45\linewidth}{\begin{prooftree}
						\AxiomC{$A,\Gamma\vdash B$}
						\AxiomC{$B,\Gamma\vdash A$}
						\RightLabel{$(\mathsf{N})$}
						\BinaryInfC{$\neg A,\Gamma\vdash \neg B$}
				\end{prooftree}}
				&
			\end{array}$}
		
		\leftline{$\begin{array}{cc}
				\parbox{.45\linewidth}{
					\begin{prooftree}
						\AxiomC{$A(a),\Gamma\vdash B$}
						\RightLabel{$(\exists\mathsf{L})$}
						\UnaryInfC{$\exists xA[a/x],\Gamma\vdash B$}
				\end{prooftree}}
			&
		\parbox{.45\linewidth}{
			\begin{prooftree}
				\AxiomC{$\Gamma\vdash A[a/t]$}
				\RightLabel{$(\exists\mathsf{R})$}
				\UnaryInfC{$\Gamma\vdash \exists xA[a/x]$}
		\end{prooftree}}\\
		\parbox{.45\linewidth}{
			\begin{prooftree}
				\AxiomC{$A[a/t],\Gamma\vdash B$}
				\RightLabel{$(\forall\mathsf{L})$}
				\UnaryInfC{$\forall xA[a/x],\Gamma\vdash B$}
		\end{prooftree}}
		&
		\parbox{.45\linewidth}{
			\begin{prooftree}
				\AxiomC{$\Gamma\vdash A(a)$}
				\RightLabel{$(\forall\mathsf{R})$}
				\UnaryInfC{$\Gamma\vdash\forall xA[a/x]$}
		\end{prooftree}}
	\end{array}$}
		In $(\exists\mathsf{R})$ and $(\forall\mathsf{L})$, $t$ is any term. In $(\exists\mathsf{L})$ and $(\forall\mathsf{R})$, an eigenvariable $a$ does not occur in $\Gamma$ and $B$.
        
		\noindent
		\underline{Additional inference rules}:
		
		\leftline{
			$\begin{array}{cc}
				\parbox{.45\linewidth}{
					\begin{prooftree}
						\AxiomC{$\Gamma\vdash A$}
						\RightLabel{$(\mathsf{ECQ})$}
						\UnaryInfC{$\Gamma, \neg A \vdash B$}
				\end{prooftree}}
				&
				\parbox{.45\linewidth}{
					\begin{prooftree}
						\AxiomC{$\Gamma, A \vdash \neg A$}
						\RightLabel{$(\mathsf{An})$}
						\UnaryInfC{$\Gamma\vdash \neg A$}
				\end{prooftree}}
			\end{array}$}

            We say a formula is \textit{derivable} in a system if there is a proof tree concluding its corresponding sequent (with an empty antecedent) using the initial sequents and inference rules of the system. A scheme is derivable if all its instances are derivable. An inference rule is said to be \textit{admissible} in a system if whenever its premises are derivable in the system, its conclusion is also derivable. We say a rule, scheme, or formula $X$ is \textit{equivalent} to $Y$ over a system if adding $X$ to the system makes $Y$ derivable or admissible, and vice versa.
	\end{dfn}

    	\begin{rem}
	\begin{itemize}
		\item The rule $(\mathsf{N})$ is equivalent to $(A \leftrightarrow B) \to (\neg A \leftrightarrow \neg B)$.
		\item The rule $(\mathsf{ECQ})$ is equivalent to $A \to (\neg A \to B)$.
		\item The rule $(\mathsf{An})$ is equivalent to $(A \to \neg A) \to \neg A$.
	\end{itemize}
        \end{rem}
        
    The following lemma is easy to prove.

	\begin{lem}\label{lem:an_equiv}
		In $\mathsf{LN}$, the Law of Non-Contradiction $\text{LNC} \equiv \neg(A \land \neg A)$ is equivalent to the rule $(\mathsf{An})$.
	\end{lem}
	
	\begin{dfn}[Arithmetical language and formulas]\label{dfn:arith_lang}
		The arithmetical language $\mathbf{NA}^{PR}$ consists of a constant $0$ (zero), function symbols for all primitive recursive functions, including the successor function $S$, and the equality symbol $=$, in addition to the logical connectives, quantifiers, and variables.
		
		Terms are constructed from variables, the constant $0$, and the function symbols. Atomic formulas are of the form $t_1 = t_2$, where $t_1$ and $t_2$ are terms.
		Formulas of arithmetic are constructed from atomic formulas, the logical connectives $\land, \lor, \to, \neg$, and the quantifiers $\forall, \exists$ as usual. 
	\end{dfn}
	
	\begin{dfn}[Arithmetic $\mathbf{NA}^{PR}$]\label{dfn:napr}
		The proof system $\mathbf{NA}^{PR}$ is obtained by adding the following initial sequents and inference rules to $\mathsf{LN}$.
		
		\noindent
		\underline{Arithmetic axioms (initial sequents)}:
		\begin{itemize}
			\item $\vdash\neg(0 = S(0))$
			\item Defining equations for all primitive recursive functions, such as:
			\begin{itemize}
				\item $\vdash\forall x (P(0) = 0)$, \quad $\vdash\forall x (P(S(x)) = x)$
				\item $\vdash\forall x (x + 0 = x)$, \quad $\vdash\forall x \forall y (x + S(y) = S(x + y))$
				\item $\vdash\forall x (x \cdot 0 = 0)$, \quad $\vdash\forall x \forall y (x \cdot S(y) = x \cdot y + x)$
				\item $\vdash\forall x (x \dotminus 0 = x)$, \quad $\vdash\forall x \forall y (x \dotminus S(y) = P(x \dotminus y))$
				\item \dots and so on.
			\end{itemize}
		\end{itemize}
		The relation $\le$ is defined as $s \le t := s \dotminus t = 0$, and the relation $<$ is defined as $s < t := S(s) \dotminus t = 0$.
		
		\noindent
		\underline{Initial sequent and Inference rules for equality}:
		
		\leftline{
			$\begin{array}{cc}
				\parbox{.45\linewidth}{
					\begin{prooftree}
						\AxiomC{}
						\RightLabel{$(\mathsf{Refl})$}
						\UnaryInfC{$\vdash t = t$}
				\end{prooftree}}
				&
				\parbox{.45\linewidth}{
					\begin{prooftree}
						\AxiomC{$\Gamma \vdash A(s)$}
						\RightLabel{$(\mathsf{Leibniz})$}
						\UnaryInfC{$s = t, \Gamma \vdash A(t)$}
				\end{prooftree}}
			\end{array}$}
		
		\noindent
		\underline{Inference rule for induction}:
		
		\leftline{
			$\begin{array}{c}
				\parbox{.9\linewidth}{
					\begin{prooftree}
						\AxiomC{$\Gamma \vdash A(0)$}
						\AxiomC{$A(a), \Gamma \vdash A(S(a))$}
						\RightLabel{$(\mathsf{Ind})$}
						\BinaryInfC{$\Gamma \vdash A(t)$}
				\end{prooftree}}
			\end{array}$}
		where the eigenvariable $a$ does not occur free in $\Gamma$.
		
		Furthermore, we define Minimal Arithmetic $\mathbf{MA}^{PR}$ as the system obtained by adding the inference rule $(\mathsf{An})$ to $\mathbf{NA}^{PR}$, and Heyting Arithmetic $\mathbf{HA}^{PR}$ as the system obtained by adding both $(\mathsf{An})$ and $(\mathsf{ECQ})$ to $\mathbf{NA}^{PR}$. We say a formula $A$ is derivable from a finite sequence of formulas $\Gamma$ in $\mathbf{NA}^{PR}$ and write $\Gamma \vdash_{\mathbf{NA}^{PR}} A$ if the sequent $\Gamma \vdash A$ is derivable in $\mathbf{NA}^{PR}$. If $\Gamma$ is empty, we simply write $\vdash_{\mathbf{NA}^{PR}} A$. The notations $\Gamma\vdash_{\mathbf{MA}^{PR}} A$ and $\Gamma\vdash_{\mathbf{HA}^{PR}} A$ are defined similarly for $\mathbf{MA}^{PR}$ and $\mathbf{HA}^{PR}$, respectively. 
	\end{dfn}
	
	\section{Basic Properties of \texorpdfstring{$\mathbf{NA}^{PR}$}{NA with PR}}\label{sec:basic-prop}
	
	Since the differences among $\mathbf{NA}^{PR}$, $\mathbf{MA}^{PR}$, and $\mathbf{HA}^{PR}$ lie only in the rules for negation, one might expect conservativity results for positive formulas. We show that $\mathbf{HA}^{PR}$ is indeed conservative over $\mathbf{NA}^{PR}$ for positive formulas. As we will see later in Section \ref{sec:sep-lang}, this result relies on the presence of primitive recursive function symbols and fails in the simpler language setting.

	\begin{prop}\label{prop:s0_0}
		$\vdash_{\mathbf{NA}^{PR}} S0=0\rightarrow A$ for any formula $A$.
	\end{prop}
	\begin{proof}
		This can be shown almost same manner as in the proof in \cite[Sec 3.2.2]{troelstra1988constructivism}) by induction on the number of connectives in the formula.
		
		(Basis). First, we consider the base case for an atomic formula $P$. We can show that $\vdash_{\mathbf{NA}^{PR}} S0=0\rightarrow a=0$ by utilizing the equality rules and induction on $a$. 
		
		(Induction step). Next, we will show that $\vdash_{\mathbf{NA}^{PR}} S0=0\rightarrow A$, assuming $\vdash_{\mathbf{NA}^{PR}} S0=0\rightarrow B$ for any formula $B$ having fewer connectives than $A$ as induction hypothesis (IH).
		
		For the case $A=\neg B$, by the induction hypothesis, we have $S0=0\vdash B$. It is easy to derive B, $S(0)=0\vdash S(0)=0$ and $S(0)=0, S(0)=0\vdash B$. By applying ($\mathsf{N}$) to them, we obtain $\neg(S0=0) S0=0 \vdash \neg B$. This yields $\vdash S0=0 \to \neg B$ by applying $(\to\mathsf{R})$ and ($\mathsf{cut}$) with axiom $\vdash\neg (S(0)=0)$.

		Other cases can be proved in the same way as \cite[Sec 3.2.2]{troelstra1988constructivism}).
	\end{proof}

	\begin{prop}[Decidability of equality]\label{prop:decidability}
		$\vdash_{\mathbf{NA}^{PR}} \forall x \forall y (x = y \lor \neg(x = y))$ holds.
	\end{prop}
	\begin{proof}
        This is proved by nested induction on $x$ and $y$. We first show $\forall x \forall y (x = y \leftrightarrow S(x) = S(y))$ by using Leibniz's law and the predecessor function $P$. Then, by applying the inference rule $(\mathsf{N})$ to this equivalence, we derive $\forall x \forall y (\neg(x = y) \leftrightarrow \neg(S(x) = S(y)))$. Subsequently, we prove $\forall y (0 = y \lor \neg(0 = y))$ by induction on $y$. Utilizing this result, the proof of $\forall x \forall y (x = y \lor \neg(x = y))$ can be completed by induction on $x$.
    \end{proof}
	
	\begin{thm}\label{thm:conserv}
		For any positive formula $P$, $\vdash_{\mathbf{HA}^{PR}} P \iff \vdash_{\mathbf{NA}^{PR}} P$.
	\end{thm}
	\begin{proof}
		We define a translation $(\cdot)^*$ for any formula $F$ such that $(\neg A)^* \equiv A^* \to (0=S(0))$, leaving other connectives and quantifiers unchanged. For a sequence of formulas $\Gamma$, $\Gamma^*$ denotes the sequence of translated formulas.
		By induction on the length of the derivation in $\mathbf{HA}^{PR}$, we show that $\Gamma \vdash_{\mathbf{HA}^{PR}} F $ yields $ \Gamma^* \vdash_{\mathbf{NA}^{PR}} F^*$. 
		
The base cases for the initial sequents and arithmetical axioms are trivial. Likewise, the cases for the structural and logical rules for connectives and quantifiers are straightforwardly handled because the translation commutes with them. For the remaining rules involving negation, $(\mathsf{ECQ})$, and $(\mathsf{An})$, we can easily show that their translated versions are admissible in $\mathbf{NA}^{PR}$. 
	
		For a positive formula $P$, its translation is identical to itself ($P^* \equiv P$) because it contains no negation. Therefore, $\vdash_{\mathbf{HA}^{PR}} P \implies \vdash_{\mathbf{NA}^{PR}} P$. The converse direction holds since $\mathbf{NA}^{PR}$ is a subsystem of $\mathbf{HA}^{PR}$.
	\end{proof}

	\begin{cor}[Trichotomy]\label{cor:trichotomy}
		$\vdash_{\mathbf{NA}^{PR}} \forall x \forall y (x < y \lor x = y \lor y < x)$ holds.
	\end{cor}
	\begin{proof}
		By Theorem \ref{thm:conserv} and $\vdash_{\mathbf{HA}^{PR}} \forall x \forall y (x < y \lor x = y \lor y < x)$, the formula is derivable.
	\end{proof}

	\section{Equivalences between Logical and Arithmetical Principles over \texorpdfstring{$\mathbf{NA}^{PR}$}{NA with PR}}\label{sec:equiv}

    In this section, we investigate the relationship between logical rules and arithmetical principles over $\mathbf{NA}^{PR}$. Specifically, we show that the rule $(\mathsf{ECQ})$ is equivalent to the definability of disjunction, and we introduce the Non-contradictory Bounded Least Number Principle ($\text{NBP}$) to characterize the rule $(\mathsf{An})$.

	\begin{lem}\label{lem:ecq_equiv}
		In $\mathbf{NA}^{PR}$, $(\mathsf{ECQ})$ is equivalent to $\neg(0=0)\to 0=S0$.
	\end{lem}
	\begin{proof}
		It is easy to see that $\neg(0=0) \to 0=S0$ is derivable using $\mathsf{ECQ}$.
		
		Conversely, assume $\neg(0=0) \to 0=S0$. We derive the rule $\mathsf{ECQ}$. Suppose we have a derivation of $\Gamma \vdash A$.
		First, from the premise $\Gamma \vdash A$ and $\vdash 0=0$, we obtain $0=0, \Gamma \vdash A$ and $\Gamma, A \vdash 0=0$ using weakening. Then, by applying the rule $(\mathsf{N})$, we can easily derive $\Gamma, \neg A \vdash \neg(0=0)$.
		Then, using the assumption, we obtain $\Gamma, \neg A \vdash 0=S0$. By Proposition \ref{prop:s0_0}, we know that $0=S0 \vdash_{\mathbf{NA}^{PR}} B$ for any formula $B$. By applying $(\mathsf{cut})$ with $\Gamma, \neg A\vdash 0=S0$ $0=S0 \vdash B$, we obtain $\Gamma, \neg A \vdash B$.
	\end{proof}

    We can show the following lemma easily.

    \begin{lem}
        $\vdash_{\mathbf{NA}^{PR}}\exists x[(x=0\rightarrow A)\land(\neg x=0\rightarrow B)]\rightarrow(A\lor B)$.
    \end{lem}
	
	\begin{prop}\label{prop:ecq_disj}
		Over $\mathbf{NA}^{PR}$, $(\mathsf{ECQ})$ is equivalent to $(A\lor B)\rightarrow\exists x[(x=0\rightarrow A)\land(\neg x=0\rightarrow B)]$.

	\end{prop}
	\begin{proof}
		Since the implication from the first clause to the second clause is easy to show, we show that the implication implies $(\mathsf{ECQ})$. Assume the implication, by Lemma \ref{lem:an_equiv}, it suffices to show $\neg0=0\rightarrow S0=0$. Since $\vdash_{\mathbf{NA}^{PR}} 0=0\lor S0=0$, by the implication, it suffices to show $\vdash_{\mathbf{NA}^{PR}} \exists x[(x=0\rightarrow 0=0)\land(\neg x=0\rightarrow S0=0)]\rightarrow \neg0=0\rightarrow S0=0\cdots (*)$.
		
		\noindent
		It is easy to see that the following sequents are derivable in $\mathbf{NA}^\text{PR}$:
		\begin{itemize}
			\item $\exists x((x=0\rightarrow0=0)\land(\neg x=0\rightarrow S0=0))\vdash \exists x(\neg x=0\rightarrow S0=0)$
			\item $a=0,\neg a=0\rightarrow S0=0\vdash\neg0=0\rightarrow S0=0$
			\item $\neg a=0,\neg a=0\rightarrow S0=0\vdash\neg 0=0\rightarrow S0=0$
		\end{itemize}
		Using these sequents and the decidability of equality $\vdash a=0\lor\neg a=0$ (Proposition \ref{prop:decidability}), we can easily derive $(*)$ by applying logical rules such as $(\lor\mathsf{L})$, $(\exists\mathsf{L})$, and $(\mathsf{cut})$.
	\end{proof}

	The definability of disjunction $(A\lor B)\leftrightarrow\exists x[(x=0\rightarrow A)\land(\neg(x=0)\rightarrow B)]$ is typically used to derive the disjunction property from the numerical existence property. However, as we will show in the subsequent sections, $\mathbf{NA}^{PR}$ is strictly weaker than $\mathbf{MA}^{PR}$, and $\mathbf{MA}^{PR}$ is strictly weaker than $\mathbf{HA}^{PR}$. Because of this weakness, this standard method cannot be directly applied in these weaker systems. Nevertheless, the following equivalence holds.

	\begin{thm}[Definability of Disjunction in $\mathbf{NA}^{PR}$]\label{thm:def_disj}
		$\vdash_{\mathbf{NA}^{PR}} A\lor B\leftrightarrow \exists x \big( (x=0 {\to} A) \land (0 < x {\to} B) \big)$ holds.
	\end{thm}
	\begin{proof}
		The proof method for this equivalence is almost identical to the method in \cite[Sec 3.3.2]{troelstra1988constructivism}). 
	\end{proof}

	As discussed in the introduction, we aim to characterize logical rules using arithmetical axioms. Our Non-contradictory Bounded Least Number Principle ($\text{NBP}$), which is provides such a characterization over subminimal systems, inspired by the classical Least Number Principle. In a subminimal setting where negation is weak, establishing the absolute non-existence of smaller witnesses can be difficult. $\text{NBP}$ asserts that it is contradictory that we have a witness $z$ for $A$, but there is not the least witness of $A$ less than or equal to $z$. Intuitively, it captures a constructive search for the smallest number satisfying $A$ in a way that avoids direct negation elimination.
	
	\begin{dfn}[Non-contradictory Bounded Least Number Principle $\text{NBP}$]\label{dfn:nbp}
		The schema $\text{NBP}$ is defined as follows: for any formula $A$
		\begin{equation*}
			 \forall z \neg \Big( A(z) \land \forall y \big(y \le z \to \neg(A(y) \land \forall u(u < y \to \neg A(u)))\big) \Big)
		\end{equation*}
	\end{dfn}

	\begin{thm}\label{thm:nbp_an}
		Over $\mathbf{NA}^{PR}$, $\text{NBP}$ is equivalent to $(\mathsf{An})$.
	\end{thm}
	\begin{proof}
		\textit{(Sketch)} By Lemma \ref{lem:an_equiv}, $(\mathsf{An})$ is equivalent to $\text{LNC} \equiv \neg(A \land \neg A)$. We show the equivalence between $\text{NBP}$ and $\text{LNC}$. If $A$ is an open formula $A(a)$ with a free variable $a$, we instantiate $\text{NBP}$ for $A(a)$ by treating $a$ as a parameter. We use the following abbreviations for the subformulas of $\text{NBP}[A(a)]$:
		
    \begin{equation*}
        B_A(a) \equiv A(a) \land \forall x(x < z \to \neg A(a))\quad C_A(b) \equiv \forall y(y \le b \to \neg B_A(y))
    \end{equation*}

		The core of the proof lies in the equivalence $A(a) \leftrightarrow C_A(0)$. Using the fact that $x < 0$ is vacuously false, we can derive $\neg B_A(0) \leftrightarrow \neg A(a)$, and subsequently $A(a) \leftrightarrow C_A(0)$. From this, it follows that $A(a) \land \neg A(a)$ is equivalent to $A(a) \land C_A(0)$ (an instance of $\text{NBP}$ at $z=0$), which establishes ($\text{NBP}\implies \text{LNC}$). For the latter half ($\text{LNC} \implies \text{NBP}$), we use induction to show $C_A(a) \vdash \neg A(a)$, and then obtain the desired result.
	\end{proof}

	Since $\mathbf{MA}^{PR}$ is defined by adding the rule $(\mathsf{An})$ to $\mathbf{NA}^{PR}$, Theorem \ref{thm:nbp_an} implies that $\text{NBP}$ characterizes $\mathbf{MA}^{PR}$ over $\mathbf{NA}^{PR}$.

\section{Separation of \texorpdfstring{$\mathbf{NA}^{PR}$}{NA with PR} and \texorpdfstring{$\mathbf{MA}^{PR}$}{MA with PR} using Models}\label{sec:sep-na-ma}

This section introduces constant domain neighbourhood models for arithmetic. Using this semantics, we construct a specific countermodel to demonstrate that the rule $(\mathsf{An})$ is not derivable, thereby strictly separating $\mathbf{NA}^{PR}$ and $\mathbf{MA}^{PR}$.

\begin{dfn}[Constant Domain neighbourhood Arithmetic Models]
    Let $\mathcal{U}(W)$ be the set of upward-closed sets of worlds. A constant Domain neighbourhood model for arithmetic is a structure $\mathcal{M} = \langle W, \leq, D, I, N \rangle$ satisfying the following conditions:
    \begin{itemize}
        \item \textbf{Worlds and Relation}: $W$ is a non-empty set of worlds, and $\leq$ is a partial order on $W$.
        \item \textbf{Domain}: $D = \mathbb{N}$ is the constant domain for all $w \in W$, where $\mathbb{N}$ is the set of the natural numbers.
        \item \textbf{Interpretation}: $I$ is an interpretation function. Constants, function symbols, and the equality symbol $=$ are interpreted strictly as standard arithmetic operations and relations on $\mathbb{N}$ in every world $w \in W$.
        \item \textbf{Negation Function}: The negation function $N: \mathcal{U}(W) \to \mathcal{U}(W)$ is a function satisfying the following condition:
        \begin{equation}\tag{P1}\label{eq:p1}
            w \in N(U) \iff w \in N(U \cap R(w))
        \end{equation}
        for any $w \in W$ and $U \in \mathcal{U}(W)$, where $R(w) = \{v \in W \mid w \leq v\}$.
    \end{itemize}
\end{dfn}

\begin{dfn}[Forcing Relation $\Vdash$]\label{dfn:forcing}
    The forcing relation $w \Vdash \phi$ for a world $w \in W$ and a formula $\phi$ (with parameters from $\mathbb{N}$) is defined inductively as follows:
    \begin{itemize}
        \item $w \Vdash t_1 = t_2$ iff $\mathbb{N} \models t_1 = t_2$, where $\mathbb{N}\models$ means the classical validity in the standard model $\mathbb{N}$. 
        \item $w \Vdash A \land B$ iff $w \Vdash A$ and $w \Vdash B$.
        \item $w \Vdash A \lor B$ iff $w \Vdash A$ or $w \Vdash B$.
        \item $w \Vdash A \to B$ iff for all $w' \geq w$, $w' \Vdash A \implies w' \Vdash B$.
        \item $w \Vdash \exists x A(x)$ iff there exists $n \in \mathbb{N}$ such that $w \Vdash A(n)$.
        \item $w \Vdash \forall x A(x)$ iff for all $w' \geq w$ and all $n \in \mathbb{N}$, $w' \Vdash A(n)$.
        \item $w \Vdash \neg A$ iff $w \in N(\{ v \in W \mid v \Vdash A \})$.
    \end{itemize}
\end{dfn}
\noindent
As is standard for Kripke semantics, the monotonicity property holds: if $w \Vdash A$ and $w \leq w'$, then $w' \Vdash A$. Consequently, the truth set $V(A) = \{ v \in W \mid v \Vdash A \}$ is always an upward-closed set, i.e., $V(A) \in \mathcal{U}(W)$.

In general, it is an open question whether the class of all constant domain neighbourhood models provides a complete semantics for the full system $\mathbf{NA}^{PR}$, or even if every such model validates all axioms of $\mathbf{NA}^{PR}$. However, we can construct a specific model $\mathcal{M}$ with finite worlds that we can strictly verify to be a model of a subsystem $\mathbf{NA}^{PR}$, and use it to separate the logical rules.

\begin{thm}[Unprovability of $(\mathsf{An})$]
$\not\vdash_{\mathbf{NA}^{PR}} (A \to \neg A) \to \neg A$.
\end{thm}
\begin{proof}
Let $\mathcal{M} = \langle W, \leq, D, I, N \rangle$ be the following model.
\begin{itemize}
    \item \textbf{Worlds and Relation}: $W = \{w_0, w_1\}$ with $w_0 \leq w_1$.
    \item \textbf{Upward-Closed Sets}: $\mathcal{U}(W) = \{\emptyset, \{w_1\}, W\}$.
    \item \textbf{Domain}: $D = \mathbb{N}$, representing the standard model of natural numbers.
    \item \textbf{Interpretation}: $I$ interprets constants, function symbols, and $=$ exactly as the standard arithmetic operations and strict identity on $\mathbb{N}$.
    \item \textbf{Negation Function $N$}: defined as $N(\emptyset) = W$, $N(\{w_1\}) = \{w_1\}$, and $N(W) = \{w_1\}$.
    \item The forcing relation $\Vdash$ is defined exactly as in Definition \ref{dfn:forcing}.
\end{itemize}

It is routine to verify that $N$ satisfies the condition \eqref{eq:p1}.

Next, we verify that all axioms and rules of $\mathbf{NA}^{PR}$ are sound in $\mathcal{M}$.
\begin{itemize}
    \item \textbf{Axioms}: The defining equations for primitive recursive functions are trivially sound. 
    \item \textbf{Rule $(\mathsf{N})$}: If $V(A) \cap R(w) = V(B) \cap R(w)$ holds from the premises, \eqref{eq:p1} implies $w \in N(V(A)) \iff w \in N(V(B))$, which preserves the validity. 
    \item \textbf{Axiom $\vdash \neg(0=S0)$}: It is sound because $\mathbb{N} \not\models 0=S0$, meaning $0=S0$ is false everywhere ($V(0=S0) = \emptyset$), yielding $V(\neg(0=S0)) = N(\emptyset) = W$. 
\end{itemize}

Before verifying the soundness of the induction schema, we establish the definability of the forcing relation in classical arithmetic (cf.\ \cite[Lem 4.14]{extended_frames}). For any formula $B(\vec{x})$ and any world $w \in W$, there exists a translation $B'_w(\vec{x})$ such that for any parameters $\vec{n} \in \mathbb{N}$:
\[ w \Vdash B(\vec{n}) \iff \mathbb{N} \models B'_w(\vec{n}) \]
We define $B'_w(\vec{x})$ by induction on the structure of $B$:
\begin{itemize}
    \item $(t_1 = t_2)'_w \equiv t_1 = t_2$.
    \item $(C \to D)'_w \equiv \bigwedge_{w' \ge w} (C'_{w'} \to D'_{w'}).$
    \item $(\forall x C)'_w \equiv \bigwedge_{w' \ge w} \forall x C'_{w'}$.
    \item $(\neg C)'_w \equiv \begin{cases} 0=0 & \text{if } w = w_1 \\ (C'_{w_0} \to 0=S0) \land (C'_{w_1} \to 0=S0) & \text{if } w = w_0 \end{cases}$
\end{itemize}
and $(\cdot)'$ commutes with all other connectives.

We verify this definability property by induction on the complexity of $B$. 

With this definability established, we verify each instance of induction: $A(0)\land \forall x(A(x)\to A(S(x)))\to \forall x A(x)$. Assume that $w\Vdash A(0)\land \forall x(A(x)\to A(S(x)))$. Then, for each $w'\geq w$ and $n\in\mathbb{N}$, $w'\Vdash A(0)$ and $w'\Vdash A(n)\Rightarrow w'\Vdash A(S(n))$ holds. This yields, by the definability of forcing, $\mathbb{N}\models A'_{w'}(0)\land \forall x(A'_{w'}(x)\to A'_{w'}(S(x)))$ and so $\mathbb{N}\models \forall xA'_{w'}(x)$. Therefore, again by the definability of forcing, $w'\Vdash\forall xA(x)$. 

Finally, we show that $(A \to \neg A) \to \neg A$ is not valid in $\mathcal{M}$ for $A \equiv \neg(0=0)$. Since $0=0$ is true in $\mathbb{N}$, $V(0=0) = W$, which yields $V(A) = N(W) = \{w_1\}$ and $V(\neg A) = N(\{w_1\}) = \{w_1\}$. To establish $w_0 \Vdash A \to \neg A$, we check all worlds accessible from $w_0$: at $w_0$, $w_0 \not\Vdash A$ holds, and at $w_1$, both $w_1 \Vdash A$ and $w_1 \Vdash \neg A$ hold. Thus, the implication $A \to \neg A$ is true at $w_0$. However, since $w_0 \notin V(\neg A)$, we have $w_0 \not\Vdash \neg A$. Therefore, $w_0 \not\Vdash (A \to \neg A) \to \neg A$, completing the proof of unprovability.
\end{proof}

\section{Separation of \texorpdfstring{$\mathbf{MA}^{PR}$}{MA with PR} and \texorpdfstring{$\mathbf{HA}^{PR}$}{HA with PR} using Models}\label{sec:sep-ma-ha}

In this section, we define models for $\mathbf{MA}^{PR}$ and construct a specific semantic countermodel to prove that $(\mathsf{ECQ})$ is not derivable in $\mathbf{MA}^{PR}$, establishing a strict separation between $\mathbf{MA}^{PR}$ and $\mathbf{HA}^{PR}$.

\begin{dfn}[Models for Minimal Arithmetic]
    A model for $\mathbf{MA}^{PR}$ is an neighbourhood model where the negation function $N$ satisfies the following additional condition: for any upward-closed set $U \in \mathcal{U}(W)$ and any world $w \in W$, $w \in N(U)$ if and only if for every world $w'$ such that $w' \geq w$, if $w' \in U$, then $w' \in N(W)$ (where $W$ is the set of all worlds).
\end{dfn}

Similarly to the previous section, it is not generally guaranteed that any neighbourhood model satisfying the above condition validates all the axioms of $\mathbf{MA}^{PR}$. However, the neighbourhood model $\mathcal{M}$ we construct below forms a model of $\mathbf{MA}^{PR}$. 

\begin{thm}[Unprovability of $\mathsf{ECQ}$]
$\not\vdash_{\mathbf{MA}^{PR}} \neg(0=0) \to 0=S(0)$ (and thus $\mathsf{ECQ}$).
\end{thm}
\begin{proof}
We construct a specific structure $\mathcal{M}$ to act as a countermodel for $\mathbf{MA}^{PR}$.
\begin{itemize}
	\item \textbf{Worlds and Relation:} $W = \{w_0\}$, with the trivial relation $w_0 \leq w_0$.
	\item \textbf{Domain:} $D = \mathbb{N}$, the standard natural numbers.
	\item \textbf{Interpretation:} Constants, function symbols, and the equality symbol $=$ are interpreted exactly as the standard arithmetic operations and strict identity on $\mathbb{N}$.
	\item \textbf{Negation Function $N$:} We define $N(U) = W$ for all $U \in \mathcal{U}(W)$. 
\end{itemize}

Next, we verify the soundness of $\mathbf{MA}^{PR}$ in $\mathcal{M}$.
Since the domain is the standard model $\mathbb{N}$ and all functions and relations are interpreted in the standard way, all standard arithmetic axioms, including the defining equations for primitive recursive functions and the induction schema, are trivially satisfied. Furthermore, our definition of the negation function $N$ ensures that the condition for $(\mathsf{An})$ holds. Thus, $\mathcal{M}$ is a valid model of $\mathbf{MA}^{PR}$.

Finally, we evaluate the formula $\neg(0=0) \to 0=S(0)$ to refute $\mathsf{ECQ}$ by Lemma \ref{lem:ecq_equiv}.
\begin{itemize}
	\item \textbf{Antecedent:} $\neg(0=0)$ evaluates to true. Since we defined the negation function such that $N(U) = W$ for all $U$, any negated formula is valid in our model.
	\item \textbf{Consequent:} $0=S(0)$ evaluates to false. Under the standard interpretation on $\mathbb{N}$, $0=S(0)$ is false.
\end{itemize}
Because the antecedent is true and the consequent is false, the implication as a whole evaluates to false. This implies that $\neg(0=0) \to 0=S(0)$ is not valid in \(\mathcal{M}\).
\end{proof}

\begin{rem}
Similarly to the previous section, because all interpretations on $\mathbb{N}$ are entirely standard, this exact model also serves to separate $\mathbf{MA}$ and $\mathbf{HA}$.
\end{rem}

\section{Separation of Arithmetical Systems with Different Languages}\label{sec:sep-lang}

To highlight the role of primitive recursive functions in weak arithmetical systems, we show that the basic systems $\mathbf{NA}$ and $\mathbf{MA}$ are strictly weaker than their primitive recursive counterparts $\mathbf{NA}^{PR}$ and $\mathbf{MA}^{PR}$.

It is well-known $\mathbf{HA}^{PR}$ is mutual interpretable with HA formulated only with $S, +, \cdot$ (cf. \cite{avigad2022mathematical}). In this section, we define $\mathbf{MA}$ and $\mathbf{NA}$ in the same simple language. Then, we show that the relationship between $\mathbf{MA}^{PR}$ and $\mathbf{MA}$, and that between $\mathbf{NA}^{PR}$ and $\mathbf{NA}$, are entirely different from the relationship between $\mathbf{HA}^{PR}$ and $\mathbf{HA}$.

\begin{dfn}[Arithmetic $\mathbf{NA}$ and $\mathbf{MA}$ in a minimal language]
    The standard language $\mathcal{L}_A$ consists of a constant symbol $0$, unary function $S$, binary functions $+$ and $\cdot$ (cf. \cite{avigad2022mathematical}).
    The proof system $\mathbf{NA}$ is obtained from $\mathsf{LN}$ by adopting the same equality rules and induction schema as $\mathbf{NA}^{PR}$, but replacing the defining equations of all primitive recursive functions with the standard basic axioms for $+, \cdot$, specifically: 
    \begin{multicols}{2}
    \begin{itemize}
        \item $\forall x \neg(0 = S(x))$
        \item $\forall x \forall y ((S(x) = S(y)) \to x = y)$
        \item $\forall x (x + 0 = x)$
        \item $\forall x \forall y (x + S(y) = S(x + y))$
        \item $\forall x (x \cdot 0 = 0)$
        \item $\forall x \forall y (x \cdot S(y) = (x \cdot y) + x)$
    \end{itemize}
    \end{multicols}
    
    Minimal Arithmetic ($\mathbf{MA}$) is defined as $\mathbf{NA} + (\mathsf{An})$.
\end{dfn}

\begin{thm}[Separation of $\mathbf{NA}^{PR}$ from $\mathbf{NA}$]
    There is a positive formula $A$ such that $\vdash_{\mathbf{NA}^{PR}} A$ and $\vdash_{\mathbf{MA}^{PR}} A$, but $\not\vdash_{\mathbf{NA}} A$ and $\not\vdash_{\mathbf{MA}} A$.
\end{thm}
\begin{proof}
    Let $A \equiv \forall x(S(x)=0 \to S(0)=0)$. In $\mathbf{NA}^{PR}$, Proposition \ref{prop:s0_0} guarantees that $\vdash_{\mathbf{NA}^{PR}} S(x)=0 \to S(0)=0$, hence $\vdash_{\mathbf{NA}^{PR}} A$. The same holds for $\mathbf{MA}^{PR}$.
    
    To show $\not\vdash_{\mathbf{MA}} A$ (and therefore $\not\vdash_{\mathbf{NA}} A$), we construct a structure $\mathcal{M}$ to act as a countermodel (cf.\ \cite{mol201808}). Because minimal logic is a subsystem of classical logic, any classical model satisfying the axioms of $\mathbf{MA}$ is a valid model.
    \begin{itemize}
        \item \textbf{Domain:} $D = \{0, 1\}$.
        \item \textbf{Constants and Successor:} $0^{\mathcal{M}} = 0$. We define $S^{\mathcal{M}}(0) = 1$ and $S^{\mathcal{M}}(1) = 0$.
        \item \textbf{Addition and Multiplication:} standard modulo 2 arithmetic.
        \item \textbf{Equality:} standard strict identity.
        \item \textbf{Negation ($\neg$):} Because minimal logic treats negation very weakly, we are free to interpret any negated formula $\neg A$ as \textbf{True}.
    \end{itemize}

    We concisely verify the soundness of $\mathbf{MA}$ in $\mathcal{M}$:
    \begin{enumerate}
        \item \textbf{Equality and Leibniz's Law:} Satisfied because $=$ is interpreted as strict identity.
        \item \textbf{Successor Injectivity:} $S$ swaps 0 and 1, making it a bijection.
        \item \textbf{Zero is not a Successor:} $\neg(0 = S(x))$ evaluates to True, since all negated formulas are True.
        \item \textbf{Ring Equations:} The modulo 2 operations satisfy the axioms for $+$ and $\cdot$.
        \item \textbf{Induction Schema:} $A(0) \land \forall x(A(x) \to A(S(x))) \to \forall x A(x)$. Any formula $A(x)$ defines a set $X=\{a \mid \mathcal{M} \models A(a)\} \subseteq \{0,1\}$. In each case, we can show the schema easily.
    \end{enumerate}

    Finally, we evaluate $A$ in the model $\mathcal{M}$. Let $x = 1$.
    \begin{itemize}
         \item \textbf{Antecedent:} $S(1) = 0$ evaluates to true. In our model, $S^{\mathcal{M}}(1)$ is $0$, which reduces to $0=0$.
        \item \textbf{Consequent:} $S(0) = 0$ evaluates to false. In our model, $S^{\mathcal{M}}(0)$ is $1$, which reduces to $1=0$.
    \end{itemize}
    Because the antecedent is true and the consequent is false, the implication as a whole evaluates to false. Thus $\not\vdash_{\mathbf{MA}} A$, and consequently $\not\vdash_{\mathbf{NA}} A$.
\end{proof}

\begin{rem}
Since the interpretation of the arithmetical symbols in $\mathcal{M}$ in \ref{sec:sep-na-ma} is exactly the same as that in the standard domain $\mathbb{N}$, this model also satisfies the axioms of $\mathbf{NA}$ restricted to the basic language $\mathcal{L}_A$. Thus, this countermodel also successfully separates $\mathbf{NA}$ and $\mathbf{MA}$.
\end{rem}

\section{Open Problems}

Our investigations also suggest further directions for future research. 

\vspace{-1em}

\paragraph*{Representability of Primitive Recursive Functions}
As demonstrated in Section \ref{sec:sep-lang}, primitive recursive function symbols play an essential role in subminimal arithmetic. While primitive recursive functions are representable in $\mathbf{HA}$, our results show that they are not representable in $\mathbf{MA}$. Consequently, the exact logical strength required for their representability must lie strictly between $\mathbf{MA}$ and $\mathbf{HA}$. This naturally raises the question: what specific logical rules or language extensions exactly characterize the representability of primitive recursive functions over weak arithmetic like $\mathbf{NA}$ or $\mathbf{MA}$?

\vspace{-1em}

\paragraph*{Robustness of Derivability}
We proved $\vdash_{\mathbf{NA}^{PR}} S0=0 \to A$ (Proposition \ref{prop:s0_0}). It is an open question how far the logic can be weakened while preserving this derivability; for instance, it remains provable even if we replace $(\mathsf{N})$ with Negative Ex Falso Quodlibet $(A \land \neg A) \to \neg B$. 

\vspace{-1em}

\paragraph*{Arithmetical Characterisations of Intuitionistic Logical Axioms}
We seek to arithmetically characterize other intuitionistic logical axioms, such as intermediate negations between $(\mathsf{N})$ and $(\mathsf{An})$ or $(\mathsf{ECQ})$ (cf.\ \cite{bsl24752}), Double Negation Elimination, and the Law of Contraposition. As an encouraging example, we recently found that the equivalence between the Law of Excluded Middle ($\text{LEM}$) and the Least Number Principle ($\text{LNP}$) over $\mathbf{HA}$ holds even over the much weaker system $\mathbf{NA}^{PR}$. Systematically analyzing such equivalences will further clarify the deep connections between logic and arithmetic.

\end{document}